%% file: main.tex
\documentclass[12pt, anon]{l4dc2024}

\title[$\widetilde{O}(T^{-1})$ Convergence to (Coarse) Correlated Equilibria in Markov Games]{$\widetilde{O}(T^{-1})$ Convergence to (Coarse) Correlated Equilibria in Full-Information General-Sum Markov Games}
\usepackage{times}
\PassOptionsToPackage{algo2e,ruled,linesnumbered,noend,noline}{algorithm2e}
\usepackage{natbib}
\usepackage{algorithm}

\input{weichao}

\newcommand{\amax}{A_{\max}}
\usepackage{mathrsfs}
\usepackage{multirow}
\usepackage{hyperref}
\usepackage{colortbl}

\newcommand{\aall}{\mca_{\operatorname{all}}}

\author{%
 \Name{Weichao Mao} \Email{weichao2@illinois.edu}\\
 \Name{Haoran Qiu} \Email{haoranq4@illinois.edu}\\
 \addr University of Illinois Urbana-Champaign
 \AND
 \Name{Chen Wang} \Email{chen.wang1@ibm.com}\\
 \Name{Hubertus Franke} \Email{frankeh@us.ibm.com}\\
 \addr IBM Research
 \AND
 \Name{Zbigniew Kalbarczyk} \Email{kalbarcz@illinois.edu}\\
 \Name{Tamer Ba\c{s}ar} \Email{basar1@illinois.edu}\\
 \addr University of Illinois Urbana-Champaign
}

\begin{document}

\maketitle

\begin{abstract}
No-regret learning has a long history of being closely connected to game theory.
Recent works have devised uncoupled no-regret learning dynamics that, when adopted by all the players in normal-form games, converge to various equilibrium solutions at a near-optimal rate of $\widetilde{O}(T^{-1})$, a significant improvement over the $O(1/\sqrt{T})$ rate of classic no-regret learners.
However, analogous convergence results are scarce in Markov games, a more generic setting that lays the foundation for multi-agent reinforcement learning.
In this work, we close this gap by showing that the optimistic-follow-the-regularized-leader (OFTRL) algorithm, together with appropriate value update procedures, can find $\widetilde{O}(T^{-1})$-approximate (coarse) correlated equilibria in full-information general-sum Markov games within $T$ iterations.
Numerical results are also included to corroborate our theoretical findings.\footnote{This preprint carries essentially the same results and title as one that was submitted to a conference on December 8, 2023, except the appearance here on page 3 of a post-submission note which provides a comparison with a recently (January 26, 2024) posted arXiv paper by \cite{cai2024near}.}
\end{abstract}


\input{1_introduction}

\input{2_preliminaries}

\input{3_ce}

\input{4_cce}

\input{5_simulations}

\section{Concluding Remarks}
In this paper, we have studied the fast convergence of no-regret learning in full-information general-sum Markov games and answer the open question of \cite{yang2022t} in the affirmative.
We have shown that within $T$ iterations, BM-OFTRL with smooth value updates finds an $\widetilde{O}(T^{-1})$-approximate CE, and OFTRL with stage-based value updates finds an $\widetilde{O}(T^{-1})$-approximate CCE, both of which match the best-known rates in normal-form games.
For future research, it would be interesting to investigate whether OFTRL with smooth value updates attains the same $\widetilde{O}(T^{-1})$ convergence to CCE as has been observed in our simulations. 
Another direction is to improve our rates in terms of the dependence on $H$ and $\log T$, or to prove any lower bounds for them.


\bibliography{ref}

\newpage
\appendix
\input{appendix_lemma}

\input{appendix_ce}
\input{appendix_cce}

\input{appendix_simulations}

\end{document}

%% file: weichao.tex
\usepackage{amsfonts}
\usepackage{xcolor}
\let\proof\relax
\let\endproof\relax
\usepackage{amsthm}
\usepackage{amsmath}
\usepackage{amssymb}
\usepackage{soul}
\usepackage{comment}
\usepackage{enumitem}
\usepackage{bbm}
\usepackage{graphicx}
\usepackage{float}
\usepackage{mathtools}
\usepackage[algo2e,ruled,linesnumbered,noend,noline]{algorithm2e}
\usepackage{bm}

\newtheorem{dfn}{Definition}

\newtheorem{thm}{Theorem}

\newtheorem{lem}{Lemma}

\newenvironment{proofsketch}{\proof}{\endproof}

\newcommand{\ra}{\rightarrow}
\newcommand{\mb}[1]{\mathbb{#1}}

\newcommand{\mc}[1]{\mathcal{#1}}

\newcommand{\rr}{\mathbb{R}}
\newcommand{\ee}{\mathbb{E}}
\newcommand{\pp}{\mathbb{P}}

\newcommand{\nn}{\mathbb{N}}
\renewcommand{\epsilon}{\varepsilon}

\newcommand{\mcs}{\mathcal{S}}
\newcommand{\mca}{\mathcal{A}}
\newcommand{\mcb}{\mathcal{B}}

\newcommand{\mcn}{\mathcal{N}}

\newcommand{\abs}[1]{\left| #1 \right|}
\newcommand{\norm}[1]{\left\| #1 \right\|}

\newcommand{\inner}[1]{\left\langle #1 \right\rangle}

\newcommand{\floor}[1]{\left\lfloor #1 \right\rfloor}
\renewcommand{\l}{\left(}
\renewcommand{\r}{\right)}
\renewcommand{\L}{\left[}
\newcommand{\R}{\right]}

\newcommand{\T}{^\top}

\newcommand\defeq{:=}

\DeclareMathOperator*{\argmax}{argmax} 

\usepackage{caption}

\usepackage{tikz}
\usepackage{xcolor}

%% file: 1_introduction.tex
\vspace{-.2em}\section{Introduction}\label{sec:intro}
Online learning has an intimate connection to game theory for finding solutions under various equilibrium concepts \citep{robinson1951iterative}. 
In no-regret learning, the learner aims to maximize its cumulative utility in response to the (possibly adversarial) outcome sequence generated by the environment. 
For normal-form games (NFGs), a folklore result states that if all the players adopt certain no-regret learning algorithms that have $O(\sqrt{T})$ regret guarantees against an adversarial environment, then they can find an $O(1/\sqrt{T})$-approximate Nash equilibrium  in two-player zero-sum games or an $O(1/\sqrt{T})$-approximate (coarse) correlated equilibrium  in general-sum games after $T$ iterations \citep{hart2000simple,cesa2006prediction}.
A broad family of no-regret learning algorithms fit into this category, including the well-known multiplicative weight updates \citep{littlestone1994weighted,freund1997decision}, follow-the-regularized/perturbed-leader \citep{kalai2005efficient}, and mirror descent \citep{nemirovskij1983problem}.

While the $O(\sqrt{T})$ regret is unimprovable against an adversarial environment, it need not be the case for learning equilibria in games because each player in a repeated game is not facing adversarial payoffs, but instead is interacting with other players who also exhibit certain learning behavior. 
Indeed, the seminal work \citep{daskalakis2011near} developed an algorithm based on the Nesterov's excessive gap technique and established its $\widetilde{O}(T^{-1})$ convergence\footnote{Throughout the paper, we use $\widetilde{O}(\cdot)$ to suppress the poly-logarithmic dependence on $T$. } to Nash equilibria (NE) in two-player zero-sum NFGs when the algorithm is adopted by both players. 
Recent works \citep{rakhlin2013optimization,syrgkanis2015fast,foster2016learning,chen2020hedging,daskalakis2021near,anagnostides2022near,anagnostides2022uncoupled} significantly strengthened this line of results by devising other no-regret learning dynamics that find different equilibrium solutions at a faster rate than $O(1/\sqrt{T})$.
Notably, \cite{syrgkanis2015fast} showed that if all the players in a general-sum NFG employ an \emph{optimistic} version of follow-the-regularizer-leader (henceforth OFTRL), the players' strategies converge to the set of coarse correlated equilibria (CCE) at a fast rate of $O(T^{-3/4})$; such a rate was later improved to $\widetilde{O}(T^{-1})$ by \cite{daskalakis2021near}. 
More recently, the $\widetilde{O}(T^{-1})$ rate was established for swap regrets and correlated equilibria (CE) in NFGs \citep{anagnostides2022near,anagnostides2022uncoupled}.

\begin{table}[!tbp]
	\centering\caption{No-regret learning convergence rates in NFGs and Markov games. 
	}\label{tbl:overview}
	\begin{tabular}{lll}
		\hline
		\multicolumn{1}{l}{Learning objective}& Normal-form games & Markov games  \\ \hline
		\begin{tabular}[]{@{}l@{}}Nash equilibrium\\ (two-player zero-sum)\end{tabular}  & $\widetilde{O}(T^{-1})$ \citep{daskalakis2011near} & $O(T^{-1})$ \citep{yang2022t} \\ \hline
		\multirow{2}{*}{\begin{tabular}[]{@{}l@{}}Correlated\\ equilibrium\end{tabular}} &  \multirow{2}{*}{$\widetilde{O}(T^{-1})$ \citep{anagnostides2022uncoupled}} & $\widetilde{O}(T^{-1/4})$ \citep{erez2022regret} \\
		&  &\cellcolor{gray!30}$\widetilde{O}(T^{-1})$ (Theorem~\ref{thm:main}) \\ \hline
		\multirow{2}{*}{\begin{tabular}[]{@{}l@{}}Coarse correlated\\ equilibrium\end{tabular}} & \multirow{2}{*}{$\widetilde{O}(T^{-1})$ \citep{daskalakis2021near}} & $\widetilde{O}(T^{-3/4})$ \citep{zhang2022policy} \\
		&  &\cellcolor{gray!30}$\widetilde{O}(T^{-1})$ (Theorem~\ref{thm:cce}) \\ \hline
	\end{tabular}
\end{table}

Despite the encouraging fast convergence results in NFGs, very few results are known for the more challenging regime of Markov games (also known as stochastic games \citep{shapley1953stochastic}). 
The only exceptions include \cite{zhang2022policy} and \cite{yang2022t}, who established the $\widetilde{O}(T^{-1})$ convergence of OFTRL (together with smooth value updates) to NE in two-player zero-sum full-information Markov games, matching the best rates in NFGs. 
As for general-sum Markov games, the best known results for CCE and CE are $\widetilde{O}(T^{-3/4})$ \citep{zhang2022policy} and $\widetilde{O}(T^{-1/4})$ \citep{erez2022regret}, respectively, which largely lag behind their $\widetilde{O}(T^{-1})$ counterparts in NFGs. 
In fact, establishing $\widetilde{O}(T^{-1})$ convergence to CCE or CE in general-sum Markov games has been raised as an important open question by \cite{yang2022t}. 

\noindent\textbf{Contributions.} In this work, we close this gap by developing no-regret learning algorithms with accompanying value update procedures and establishing their fast $\widetilde{O}(T^{-1})$ convergence to CCE or CE in general-sum Markov games. 
For CE (Section~\ref{sec:ce}), we consider the OFTRL algorithm with a log-barrier regularizer, and integrate it with the celebrated external-to-swap-regret reduction \citep{blum2007external} and smooth value updates.
Our $\widetilde{O}(T^{-1})$ convergence analysis builds on a Regret bounded by Variation in Utilities (RVU) property \citep{syrgkanis2015fast} for the weighted swap regret at each state. 
We make a seemingly trivial observation that swap regrets are always non-negative and use it to easily bound the second-order path lengths of the learning dynamics.
For CCE (Section~\ref{sec:cce}), we consider standard OFTRL with negative entropy regularization but combine it with a stage-based value update scheme. 
We show that this algorithm induces a no-\emph{average}-regret problem within each stage, which allows us to apply existing analysis for the \emph{individual} regret of the players \citep{daskalakis2021near}. 
Table~\ref{tbl:overview} compares our results with the best-known convergence rates of no-regret learning in NFGs and Markov games. 
We further provide numerical results (Section~\ref{sec:simulations}) to validate the $\widetilde{O}(T^{-1})$ convergence behavior of our algorithms. 

\noindent\textbf{Post conference submission note.} An independent paper by \cite{cai2024near} investigates the same problem of $\widetilde{O}(T^{-1})$ convergence to CE in general-sum Markov games. They also consider OFTRL with the log-barrier regularizer and smooth value updates and establish a similar convergence rate as in Theorem~\ref{thm:main} of this paper. Interestingly, their algorithm performs V-value updates, which allows for a more preferred decentralized implementation, and they prove its equivalence to Q-value updates as we had done in this paper. Compared with their convergence rate for CE, however, our Theorem~\ref{thm:main} happens to shave off an additional $\log T$ factor by using a more refined analysis of a weighted average of a sequence (Lemma 4 from \cite{yang2022t}). Our work has further established the $\widetilde{O}(T^{-1})$ convergence to CCE (Section~\ref{sec:cce}), which has not been considered by \cite{cai2024near}.


\noindent\textbf{Further related work.} Learning in Markov games has been widely studied in multi-agent reinforcement learning (MARL) \citep{bai2020provable,wei2021last,cen2021fast,cen2022faster,zhao2022provably,cai2023uncoupled,wang2023breaking,mao2023multi}. 
Closest to ours are the works by \cite{liu2020sharp,jin2022v,songcan,mao2022improving,mao2022provably} who have shown $\widetilde{O}(1/\sqrt{T})$ convergences to CCE or CE in multi-player general-sum Markov games under bandit feedback.

%% file: 2_preliminaries.tex
\vspace{-0.2cm}\section{Preliminaries}\label{sec:preliminaries}
\textbf{No-regret learning.} Let $\mca$ be a finite set of actions. 
At each iteration $t\in\nn_+$, a learning agent selects a strategy $\bm{x}^t\in\Delta(\mca)$ as a probability distribution over the action space. 
The environment returns a utility vector $\bm{u}^t\in\rr^{|\mca|}$, and the agent obtains a utility of $\langle \bm{x}^t,\bm{u}^t\rangle$. 
The classic notion of regret, or more generally $\Phi$-regret \citep{greenwald2003general}, is used to measure the performance of the learning agent in terms of the suboptimality in hindsight. 
Formally, given a sequence of strategies $(\bm{x}^1,\dots, \bm{x}^T)$ over $T$ iterations, the incurred $\Phi$-regret is defined as
\begingroup
\setlength{\abovedisplayskip}{4pt}
\setlength{\belowdisplayskip}{4pt}
\setlength{\abovedisplayshortskip}{4pt}
\setlength{\belowdisplayshortskip}{4pt}
\begin{equation}\label{eqn:phi_regret}
	\operatorname{Reg}^T_{\Phi} \defeq \max_{\phi^*\in\Phi} \sum_{t=1}^T \inner{\phi^*(\bm{x}^t) - \bm{x}^t, \bm{u}^t}.
\end{equation}
\endgroup 
In particular, $\operatorname{Reg}^T_{\Phi}$ is called \emph{external regret} (or simply \emph{regret}) if $\Phi$ is the set of all constant functions $\{\phi: \phi(\bm{x}) =\phi(\bm{x}'), \forall \bm{x},\bm{x}'\in\Delta(\mca)\}$, and \emph{swap regret} if $\Phi$ is the set of all linear transformations $\{\phi: \phi(\bm{x}) = Q\T \bm{x}, \text{ where } Q \text{ is a row stochastic matrix}\}$. 
One can see that swap regret is a more powerful notion of hindsight rationality by allowing a broader class of possible deviations. 

\noindent\textbf{Markov game.} An $N$-player episodic Markov game is defined by a tuple $\mb{G}=(\mc{N}, H, \mc{S},\allowbreak \{\mc{A}_i\}_{i=1}^N,\allowbreak \{r_i\}_{i=1}^N, P)$, where (1) $\mc{N} = \{1,2,\dots,N\}$ is the set of agents; (2) $H\in\mb{N}_+$ is the number of time steps in each episode; (3) $\mc{S}$ is the finite state space; (4) $\mc{A}_i$ is the finite action space for agent $i\in\mc{N}$; (5) $r_i:[H]\times \mc{S}\times \aall \ra [0,1]$ is the reward function for agent $i$, where $\aall = \bigtimes_{i=1}^N \mc{A}_i$ is the joint action space; and (6) $P: [H]\times \mc{S}\times \aall\ra \Delta(\mc{S})$ is the state transition function. 
The agents interact in an unknown environment for $T$ episodes. 
At each step $h\in[H]$, the agents observe the state $s_h \in \mc{S}$ and take actions $a_{h,i} \in\mc{A}_i$ simultaneously.  
Agent $i$ then receives its reward $r_{h,i}(s_h,\bm{a}_h)$, where $\bm{a}_h =  (a_{h,1},\dots, a_{h,N}) \in\aall$, and the environment transitions to the next state $s_{h+1}\sim P_h(\cdot | s_h,\bm{a}_h)$.
We make a standard assumption \citep{jin2022v,songcan} that each episode starts from a fixed initial state $s_1$. 
Let $S = |\mc{S}|$, $A_i = |\mc{A}_i|$, and $A_{\max} = \max_{i\in\mc{N}}A_i$.

\noindent\textbf{Policy and value function.} A (Markov) policy $\pi_i\in\Pi_i:[H]\times \mc{S}\ra \Delta(\mc{A}_i)$ for agent $i\in\mc{N}$ is a mapping from the time index and state space to a distribution over its own action space. 
Each agent seeks to find a policy that maximizes its own cumulative reward. 
A joint, product policy $\pi = (\pi_1,\dots,\pi_N)\in\Pi$ induces a probability measure over the sequence of states and joint actions. 
We use the subscript $-i$ to denote the set of agents excluding agent $i$, i.e., $\mc{N}\backslash \{i\}$. 
We can rewrite $\pi = (\pi_i, \pi_{-i})$ using this convention. For a joint policy $\pi$, and for any $h\in[H]$, $s \in \mc{S}$, and $\bm{a}\in \aall$, we define the value function and Q-function for agent $i$ as 
\begingroup
\setlength{\abovedisplayskip}{4pt}
\setlength{\belowdisplayskip}{4pt}
\setlength{\abovedisplayshortskip}{4pt}
\setlength{\belowdisplayshortskip}{4pt}
\begin{equation}\label{eqn:value_fn}
\begin{aligned}
	\scriptsize
	V_{h,i}^\pi (s)= \ee \bigg[\hspace{-.3em} \sum_{h'=h}^{H} \hspace{-.3em}r_{h',i}(s_{h'},\bm{a}_{h'})| s_h = s \bigg],Q_{h,i}^\pi (s,\bm{a}) = \ee \bigg[ \hspace{-.3em}\sum_{h'=h}^{H}\hspace{-.3em}r_{h',i}(s_{h'},\bm{a}_{h'})| s_h = s, \bm{a}_h = \bm{a} \bigg].	
\end{aligned}
\end{equation}
\endgroup
For notational convenience, for any $V:\mcs\ra\rr$, we define $
\left[{P}_h V\right](s, \bm{a})\defeq\mathbb{E}_{s^{\prime} \sim {P}_h(\cdot \mid s, \bm{a})}\left[V\left(s^{\prime}\right)\right].
$
For an arbitrary Q-function $Q_{h,i}:\mcs\times\aall\ra\rr$, we write $[Q_{h,i}\pi_h](s) \defeq \inner{Q_{h,i}(s,\cdot), \pi_h(s,\cdot)}$ and $[Q_{h,i}\pi_{h,-i}](s,a_i) \defeq \inner{Q_{h,i}(s,a_i,\cdot), \pi_{h,-i}(s,\cdot)}$ for short.


\noindent\textbf{Full-information feedback.} Following \cite{zhang2022policy,yang2022t,erez2022regret}, we consider the full-information feedback setting where each agent can observe the expected rewards it would have received had it played any candidate action. 
In our formulation, this can be interpreted as an oracle from which each agent $i$ can query $[Q_{h,i}\pi_{h,-i}](s,a_i)$ for each candidate action $a_i\in\mca_i$ at any state $s\in\mcs$. 

\noindent\textbf{Correlated policy and (coarse) correlated equilibrium.} We define $\pi =\{\pi_{h}: \rr \times (\mc{S}\times \mc{A})^{h-1}\times\mc{S}\ra \Delta(\mc{A})\}_{h\in[H]}$ as a (non-Markov) \emph{correlated policy}. 
Specifically, 
the agents first sample the value of $z\in\rr$ from an underlying distribution using a common source of randomness (e.g., a common random seed), and then they can use $z$ to coordinate their choices of actions. 
Such a virtual coordinator is crucial for learning CCE/CE  as shown in the literature \citep{mao2022provably,songcan,jin2022v}, and is standard in decentralized learning \citep{bernstein2009policy,arabneydi2015reinforcement,mao2020information}.
For each $h\in[H]$, $\pi_{h}$ maps from $z\in\rr$ and a state-action history $(s_1, \bm{a}_1,\dots, s_{h-1},\bm{a}_{h-1}, s_h)$ to a distribution over the joint action space $\aall$. 
Unlike product policies, such action distributions in general cannot be factorized into independent probability distributions over the individual action spaces. 
For a correlated policy $\pi$, we let $\pi_i$ and $\pi_{-i}$ be the proper marginal distributions of $\pi$ whose outputs are restricted to $\Delta(\mc{A}_i)$ and $\Delta(\mc{A}_{-i})$, respectively. 
The value functions for non-Markov correlated policies at step $h=1$ are defined similarly as those for product policies \eqref{eqn:value_fn}.

For any correlated policy $\pi = (\pi_i,\pi_{-i})$, the best response value of agent $i$ is denoted by 
$
V_{1,i}^{\dagger, \pi_{-i}}(s_1)  \defeq \sup_{\pi_i^\dagger}  V_{1,i}^{\pi_i^{\dagger}, \pi_{-i}}(s_1), 
$
where the supremum is taken over all (non-Markov) policies of agent $i$ independent of the randomness of $\pi_{-i}$. 
A policy $\pi_i^{\dagger}$ is agent $i$'s best response to $\pi_{-i}$ if it achieves the supremum. 
Given the PPAD-hardness of Nash equilibria~\citep{daskalakis2009complexity}, people often study relaxed equilibrium concepts in general-sum games, such as CCE and CE. 
\begin{dfn} ($\epsilon$-CCE)
	For any $\epsilon>0$, a correlated policy $\pi=(\pi_i,\pi_{-i})$ is an $\epsilon$-approximate coarse correlated equilibrium if $V_{1,i}^{\pi_i,\pi_{-i}}(s_1)\geq V_{1,i}^{\dagger, \pi_{-i}}(s_1) - \epsilon,\forall i\in\mc{N}.$
\end{dfn}
To properly define a CE, we need to first specify the concept of a strategy modification. 
Formally, for agent $i$, a strategy modification $\phi_i = \{\phi_{h,i}^s:h\in[H],s\in\mc{S}\}$ is a set of mappings from agent $i$'s action space to itself, i.e., $\phi_{h,i}^s:\mc{A}_i\ra \mc{A}_i$. 
Given a strategy modification $\phi_i$, whenever a policy $\pi$ selects the joint action $\bm{a} = (a_{1},\dots,a_{N})$ at step $h$ and state $s$, the modified policy $\phi_i \diamond \pi$ will select $(a_{1},\dots,a_{i-1}, \phi_{h,i}^s(a_{i}),a_{i+1},\dots,a_{N})$ instead. 
Let $\Phi_i$ denote the set of all possible strategy modifications for agent $i$. 
A CE states that no agent has the incentive to deviate from a correlated policy $\pi$ by using any strategy modification. 
\begin{dfn}($\epsilon$-CE)
	For any $\epsilon>0$, a correlated policy $\pi$ is an $\epsilon$-approximate correlated equilibrium if 
	$
	V_{1,i}^{\pi}(s_1) \geq \max_{\phi_i\in \Phi_i} V_{1,i}^{\phi_i\diamond \pi}(s_1) - \epsilon,\forall i\in\mc{N}.
	$
\end{dfn}

%% file: 3_ce.tex
\vspace{-1em}\section{Convergence to Correlated Equilibria}\label{sec:ce}

In this section, we present our optimistic follow-the-regularized-leader (OFTRL) algorithm for learning correlated equilibria in general-sum Markov games in Section~\ref{subsec:ce_algorithm}, and then establish its $\widetilde{O}(T^{-1})$ convergence in Section~\ref{subsec:ce_analysis}. 

\vspace{-.5em}\subsection{Algorithm}\label{subsec:ce_algorithm}
Algorithm~\ref{alg:oftrl} describes the OFTRL procedure run by agent $i\in\mcn$.
Since the algorithms run by all the agents are exactly symmetric, in the following, we only illustrate our algorithm using a single agent $i$ as an example. 
Algorithm~\ref{alg:oftrl} consists of three major components: The policy update step that computes the strategy for each matrix game, the value update step that updates the (Q-)value functions, and the policy output step that generates a CE policy.

\begin{algorithm}[!t]
	\textbf{Initialize:} $Q_{h,i}^0(s,\bm{a})\gets 0,\pi_{h,i}^0(s,a_i)\gets 1/A_i, \forall s\in\mcs,h\in[H],a_i,a_i'\in\mca_i,\bm{a}\in\aall$;
	
	\For{iteration $t\gets 1$ to $T$}
	{
		\textbf{Policy update:} 
		
		\For{action $a_i\in\mca_i$}
		{
			$\ell_{h,i}^{t,a_i}(s,a_i')\gets \sum_{j=1}^{t-1} w_j \pi_{h,i}^{j}(s, a_i) [Q_{h,i}^j \pi_{h,-i}^{j}](s,a_i') + w_t \pi_{h,i}^{t-1}(s, a_i)[Q_{h,i}^{t-1} \pi_{h,-i}^{t-1}](s,a_i')$;
			
			$q_{h,i}^{t, a_i}(s,\cdot) \gets \argmax_{\bm{x}\in\Delta(\mca_i)}\l\langle\bm{x}, \eta\ell_{h,i}^{t,a_i}(s,\cdot)/w_t\rangle - \mc{R}(\bm{x})\r, \forall s\in\mcs,h\in[H]$;
		}
		Find $\pi_{h,i}^t$ such that $\pi_{h,i}^t(s,\cdot) = \sum_{a_i\in\mca_i} \pi_{h,i}^t(s,a_i) q_{h,i}^{t,a_i}(s, \cdot), \forall s\in\mcs,h\in[H], a_i\in\mca_i$;
		
		\textbf{Value update:}  
		
		\For{$h\gets H$ to $1$}
		{
			$Q^{t}_{h,i}(s,\bm{a})\gets (1-\alpha_t) Q_{h,i}^{t-1} (s,\bm{a}) + \alpha_t \l r_{h,i} + P_h[Q_{h+1,i}^t\pi_{h+1}^t]\r(s,\bm{a}),\forall s\in\mcs,\bm{a}\in\aall$;
		}
	}
	\textbf{Output policy:} $\bar{\pi} = \bar{\pi}_1^T$, where $\bar{\pi}_h^t$ is defined in Algorithm~\ref{alg:certify}.
	
	\caption{Optimistic follow-the-regularized-leader for correlated equilibria (agent $i$)}\label{alg:oftrl}
\end{algorithm}

\noindent\textbf{Policy update.} At each fixed $(s,h)\in\mcs\times[H]$, the agents are essentially faced with a sequence of matrix games, where the payoff matrix for agent $i$ in the $t$-th matrix game is given by the estimated Q-function $Q_{h,i}^t(s, \cdot)$ at the corresponding iteration $t$.
For learning CE in matrix games, a folklore result suggests that each agent should employ a no-swap-regret learning algorithm. 
Specifically, suppose that each agent employs a no-swap-regret algorithm such that the cumulative swap regret up to time $T\in\nn_+$ is upper bounded by $\operatorname{SwapReg}^T$; then, the empirical distribution of the joint actions played by the players is an $(\operatorname{SwapReg}^T/T)$-approximate CE \citep{hart2000simple}. 

For a fixed matrix game at $(s,h)\times\mcs\times[H]$, we follow the generic reduction introduced in \citep{blum2007external} to obtain a no-swap-regret learning algorithm $\mathscr{A}_{\text{swap}}$ from a no-(external-)regret base algorithm $\mathscr{A}$. 
Specifically, \cite{blum2007external} maintain a separate no-regret algorithm $\mathscr{A}_a$ for each candidate action $a\in\mca_i$ of the agent. 
$\mathscr{A}_{\text{swap}}$ computes a strategy by combining the strategies of the $A_i$ base algorithms. 
At time step $t\in[T]$, each base algorithm $\mathscr{A}_a$ outputs a distribution $q^{t, a}(\cdot)\in\Delta(\mca_i)$, where $q^{t,a}(a')$ is the probability that it selects $a'\in\mca_i$. 
Then, a (row) stochastic matrix $q^t \in \rr^{A_i \times A_i}$ is constructed, where the $a$-th row of $q^t$ is equal to the $q^{t,a}$ vector. 
$\mathscr{A}_{\text{swap}}$ obtains the action selection strategy by computing a stationary distribution\footnote{It is known that such a distribution $\pi^t$ exists and is computationally efficient.} $\pi^{t}\in\Delta(\mca_i)$ of $q^t$ such that $(q^t)\T \pi^t = \pi^t$. 
Upon receiving the payoff vector $\bm{u}^t\in\rr^{A_i}$ (in the case of Algorithm~\ref{alg:oftrl}, $\bm{u}^t = [Q_{h,i}^j \pi_{h,-i}^j](s, \cdot)$ for agent $i$) from the environment, $\mathscr{A}_{\text{swap}}$ returns to each $\mathscr{A}_a$ base algorithm a $\pi^t(a)$ fraction of the received utility, so that $\mathscr{A}_a$ is updated with a utility vector of $\pi^t(a)\bm{u}^t \in\rr^{A_i}$. 
It is shown that $\mathscr{A}_{\text{swap}}$ guarantees no-swap-regret as long as each base algorithm $\mathscr{A}_a$ has sublinear (external) regret in $T$. 

In Algorithm~\ref{alg:oftrl}, we use weighted OFTRL as the no-regret base algorithm $\mathscr{A}$. 
OFTRL \citep{syrgkanis2015fast} extends the standard FTRL paradigm by maintaining a prediction sequence $\bm{m}^t$ of the utilities. 
Given a utility sequence $(\bm{u}^1, \dots,\bm{u}^T)$, OFTRL computes the strategies by
\begingroup
\setlength{\abovedisplayskip}{4pt}
\setlength{\belowdisplayskip}{4pt}
\setlength{\abovedisplayshortskip}{4pt}
\setlength{\belowdisplayshortskip}{4pt}
\begin{equation}\label{eqn:oftrl}
\bm{x}^t \defeq \argmax_{\bm{x}\in\Delta(\mca_i)} \bigg\{ \eta\Big\langle\bm{x},\bm{m}^t + \sum_{j=1}^{t-1} \bm{u}^j\Big\rangle - \mc{R}(\bm{x}) \bigg\},	
\end{equation}
\endgroup
where $\eta>0$ is the learning rate, and $\mc{R}$ is the regularizer. 
In Algorithm~\ref{alg:oftrl}, we instantiate \eqref{eqn:oftrl} with $\bm{m}^t = \bm{u}^{t-1}$ and the log-barrier regularizer $\mc{R}(\bm{x}) = -\sum_{a_i \in \mca_i} \log(\bm{x}[a_i])$. 
Such a log-barrier regularizer satisfies the self-concordant condition in \cite{anagnostides2022uncoupled}, which is used to establish the Regret bounded by Variation in Utilities (RVU) property \citep{syrgkanis2015fast} of the swap regret. 
Due to the time-varying learning rates in the value update step (to be discussed momentarily), we additionally use a weighted variant of OFTRL that considers a weighted sum over the utility sequence. 
The choice of the weights $\{w_j\}_{j\in[t]}$ will also be defined shortly. 
Combining the OFTRL base algorithm, the utility weights and the external-to-swap-regret reduction, we arrive at the policy update rule as presented in Algorithm~\ref{alg:oftrl}. 
With the \cite{blum2007external} reduction, we name our no-swap-regret algorithm BM-OFTRL. 

\noindent\textbf{Value update.} For any $(h, s, \bm{a})$, we update the Q-value estimates at each iteration in a Bellman manner using a weighted average of previous estimates. 
We perform incremental updates using the classic step size $\alpha_t = (H+1)/(H+t)$ proposed by \cite{jin2018q}. 
With this step size, the value update rule in Algorithm~\ref{alg:oftrl} effectively becomes:
\begingroup
\setlength{\abovedisplayskip}{4pt}
\setlength{\belowdisplayskip}{4pt}
\setlength{\abovedisplayshortskip}{4pt}
\setlength{\belowdisplayshortskip}{4pt}
\begin{equation}\label{eqn:value_updates}
	Q^{t}_{h,i}(s,\bm{a})= \sum_{j=1}^t \alpha_t^j \l r_{h,i} + P_h[Q_{h+1,i}^j\pi_{h+1}^j]\r(s,\bm{a}),\forall s\in\mcs,\bm{a}\in\aall,
\end{equation}
\endgroup
where $\alpha_t^j \defeq \alpha_j \prod_{j'=j+1}^t (1-\alpha_{j'})$ and $\alpha_t^t \defeq \alpha_t$. 
One can verify that $\sum_{j=1}^t \alpha_t^j = 1$. 
Given the time-varying weights $\alpha_t^j$, to ensure that our policy update step is no-swap-regret in the matrix games defined by the Q-value estimates, we define the weights of our weighted OFTRL procedure in Algorithm~\ref{alg:oftrl} to be $w_j \defeq \alpha_t^j / \alpha_t^1$ for any fixed $t\in[T]$. 

\begin{algorithm}[!t]
	\textbf{Input:} Policy trajectory $\{\pi_h^t\}_{h\in[H],t\in[T]}$ of Algorithm~\ref{alg:oftrl};
	
	\For{step $h'\gets h$ to $H$}
	{
		Sample $\tau\in[t]$ with probability $\pp(\tau=j) = \alpha_t^j$;
		
		Play policy $\pi_{h'}^\tau$ at step $h'$;
		
		Set $t \gets \tau$.
	}
	
	\caption{Policy $\bar{\pi}_h^t$}\label{alg:certify}
\end{algorithm}

\noindent\textbf{Policy output.} Our output policy $\bar{\pi}$ is a state-wise weighted average of the history policies, where the weights are again related to the step sizes $\alpha_t^j$. 
The construction of $\bar{\pi}$ is formally defined in Algorithm~\ref{alg:certify}, which is closely related to the ``certified policies'' from \cite{bai2020near}. 
Specifically, Algorithm~\ref{alg:certify} takes the policy trajectory $\{\pi_h^t\}_{h\in[H], t\in[T]}$ of Algorithm~\ref{alg:oftrl} as input. 
For each step $h\in[H]$, Algorithm~\ref{alg:certify} randomly samples a joint policy from the policy trajectory using the sampling probabilities $\alpha_t^j$ and let all the agents play this joint policy at the given step. 
Similar to \cite{songcan,mao2022improving,zhang2022policy}, the constructed policy $\bar{\pi}$ is a correlated policy because the agents implicitly use a common source of randomness to select the same history iteration.
We will show that the output policy constitutes an approximate CE. 
We also remark that existing results for learning Nash equilibria in two-player zero-sum Markov games \citep{zhang2022policy,yang2022t} do not require such a shared randomness and generally output Markov policies.

\vspace{-.5em}\subsection{Analysis}\label{subsec:ce_analysis}
In the following, we present the analysis of Algorithm~\ref{alg:oftrl}. 
We use the following notion of $\operatorname{CE-Gap}$ to measure the distance of a correlated policy to a CE:
\begingroup
\setlength{\abovedisplayskip}{6pt}
\setlength{\belowdisplayskip}{6pt}
\setlength{\abovedisplayshortskip}{6pt}
\setlength{\belowdisplayshortskip}{6pt}
\[
\operatorname{CE-Gap}(\pi) \defeq  \max_{i\in\mcn}\max_{\phi_i\in \Phi_i} \l V_{1,i}^{\phi_i\diamond \pi}(s_1) - V_{1,i}^{\pi}(s_1)\r, 
\] 
\endgroup
where recall that $\Phi_i$ is the set of strategy modifications for agent $i$. 
The following theorem states that Algorithm~\ref{alg:oftrl} finds an $\widetilde{O}(T^{-1})$-approximate CE in $T$ iterations. 
\begin{thm}\label{thm:main}
	If Algorithm~\ref{alg:oftrl} is run on an $N$-player episodic Markov game for $T$ iterations with a learning rate $\eta = \frac{1}{256 NH \sqrt{H\amax}}$, the output policy $\bar{\pi}$ satisfies:
	\begingroup
	\setlength{\abovedisplayskip}{5pt}
	\setlength{\belowdisplayskip}{5pt}
	\setlength{\abovedisplayshortskip}{5pt}
	\setlength{\belowdisplayshortskip}{5pt}
	\[
	\operatorname{CE-Gap}(\bar{\pi}) \leq \frac{6144 N H^{\frac{7}{2}} \amax^{\frac{5}{2}}\log T}{T}. 
	\]
	\endgroup
\end{thm}
Theorem~\ref{thm:main} improves the existing $\widetilde{O}(T^{-1/4})$ rate \citep{erez2022regret} of no-regret learning to CE in full-information Markov games. 
The parameter dependences in Theorem~\ref{thm:main} also match the best known rate for normal-form games \citep{anagnostides2022uncoupled}, except that Theorem~\ref{thm:main} introduces an additional $O(H^{\frac{7}{2}})$ dependence on the Markov game episode length. 
We remark that we make no effort to improve the constant factors in the bounds, which can certainly be tightened.

The proof structure of Theorem~\ref{thm:main} is conceptually similar to those for learning Nash equilibria in two-player zero-sum Markov games \citep{zhang2022policy,yang2022t} .
We first introduce a few notations to facilitate the proof.
For any $(s,h)\in\mcs\times[H]$, we define the per-state weighted swap regret up to iteration $t\in[T]$ in the corresponding matrix game as
\begingroup
\setlength{\abovedisplayskip}{5pt}
\setlength{\belowdisplayskip}{5pt}
\setlength{\abovedisplayshortskip}{5pt}
\setlength{\belowdisplayshortskip}{5pt}
\[
\begin{aligned}
	\operatorname{SwapReg}_{h,i}^t(s) &\defeq \max_{\phi_{h,i}^s:\mca_i\ra\mca_i}\sum_{j=1}^t \alpha_t^j\inner{\phi_{h,i}^s\diamond \pi_{h,i}^j(s,\cdot) - \pi_{h,i}^j(s,\cdot), [Q_{h,i}^j \pi_{h,-i}^j](s, \cdot)},\\
	\operatorname{SwapReg}_{h}^t &\defeq \max_{i\in\mcn}\max_{s\in\mcs}\operatorname{SwapReg}_{h,i}^t(s).
\end{aligned}
\]
For any $(h,t)\in[H]\times[T]$, we further define the best response CE value gap as
\[
\delta_h^t \defeq \max_{i\in\mcn}\max_{\phi_i}\max_{s\in\mcs} \l V_{h,i}^{\phi_i\diamond \bar{\pi}_h^t}(s) - V_{h,i}^{\bar{\pi}_h^t}(s)\r,
\]
\endgroup 
where $\bar{\pi}_h^t$ is defined in Algorithm~\ref{alg:certify} and we slightly abuse the notation $\phi_i$ to denote a strategy modification that is only effective starting from step $h$. 
By the definition of $\delta_h^t$ and $\bar{\pi}$, one can easily see that $\operatorname{CE-Gap}(\bar{\pi})= \operatorname{CE-Gap}(\bar{\pi}_1^T) \leq \delta_1^T$. 
To control $\delta_1^T$, we first use the following lemma to establish the recursive relationship of the best response CE value gaps between two consecutive steps $h$ and $h+1$:
\begin{lem}\label{lemma:recursion}
	(Recursion of best response CE value gaps) For any fixed $(h,t)\in[H]\times[T]$, we have
	\begingroup
	\setlength{\abovedisplayskip}{0pt}
	\setlength{\belowdisplayskip}{0pt}
	\setlength{\abovedisplayshortskip}{0pt}
	\setlength{\belowdisplayshortskip}{0pt}
	\begin{equation}\label{eqn:a2}
		\delta_h^t \leq \sum_{j=1}^t \alpha_t^j \delta_{h+1}^j + \operatorname{SwapReg}_h^t.
	\end{equation}
	\endgroup 
\end{lem}
Therefore, upper bounding $\operatorname{CE-Gap}(\bar{\pi})$ breaks down to controlling the per-state weighted swap regrets  for every $(s,h)\in\mcs\times[H]$. 
We can further establish the upper bound of  $\operatorname{SwapReg}_{h,i}^t(s)$ in the next lemma. 
The proof of this lemma relies on an RVU bound for the swap regret of BM-OFTRL under time-varying learning rates in normal-form games. 
\begin{lem}\label{lemma:perstate_weighted_regret}
	(Per-state weighted swap regret bounds) For any $t\in[T], h\in[H], s\in\mcs$ and $i\in\mcn$, Algorithm~\ref{alg:oftrl} ensures that
	\small 
	\begingroup
	\setlength{\abovedisplayskip}{0pt}
	\setlength{\belowdisplayskip}{0pt}
	\setlength{\abovedisplayshortskip}{0pt}
	\setlength{\belowdisplayshortskip}{0pt}
	\begin{align}
		\operatorname{SwapReg}_{h,i}^t(s) \leq& \frac{4A_i^2 H\log t}{\eta t}  +  \frac{32\eta H^3N^2}{t} + 8\eta N H^2\sum_{j=2}^t\sum_{k\neq i}\alpha_t^j \Big\|\pi_{h,k}^j(s,\cdot) - \pi_{h,k}^{j-1}(s,\cdot)\Big\|_1^2.\label{eqn:a1}
	\end{align}
\endgroup
	\normalsize If $\eta \leq \frac{1}{256NH \sqrt{H\amax}}$, we further have \small
	\begingroup
	\setlength{\abovedisplayskip}{0pt}
	\setlength{\belowdisplayskip}{0pt}
	\setlength{\abovedisplayshortskip}{0pt}
	\setlength{\belowdisplayshortskip}{0pt}
	\begin{align}
		\sum_{i=1}^N\operatorname{SwapReg}_{h,i}^t(s) \leq &\frac{4 N\amax^2 H\log t}{\eta t}  + \frac{32\eta NH^2(N^2+H)}{t}\nonumber\\
		&- \frac{1}{2048\eta H }\sum_{i=1}^N\sum_{j=2}^t  \frac{\alpha_t^j}{A_i}\norm{\pi_{h,i}^{j}(s,\cdot) - \pi_{h,i}^{j-1}(s,\cdot)}_1^2.\label{eqn:a0}
	\end{align}
\endgroup
\end{lem}

We note that there is a discrepancy between \eqref{eqn:a2} and \eqref{eqn:a0}. 
Specifically, \eqref{eqn:a2} requires an upper bound for the \emph{maximum} of the swap regrets over the agents while \eqref{eqn:a0} controls the \emph{sum} of them. 
This poses some additional challenges for learning NE (in zero-sum Markov games) or CCE in existing works \citep{zhang2022policy,yang2022t}, because some players may experience negative regret \citep{hsieh2021adaptive} and the sum of regrets in general does not upper bound the maximum individual regret of the players. 
For CE, however, we can take advantage of a seemingly trivial property that the swap regret is always non-negative.
This is in sharp contrast to the (external) regret and one can easily verify this property by letting all the strategy modifications $\phi_{h,i}^s$ in $\operatorname{SwapReg}_{h,i}^t(s)$ be identity mappings. 
In this case, the discrepancy will not impede us as we can easily upper bound the maximum \eqref{eqn:a2} by the sum \eqref{eqn:a0},
which already yields an $\widetilde{O}(t^{-1})$ convergence rate. Our proof of Theorem~\ref{thm:main} 
instead follows a different route that upper bounds the second-order path lengths of the learning dynamics, which leads to an improved rate in terms of the dependence on $N$. 

%% file: 4_cce.tex
\vspace{-.5em}\section{Convergence to Coarse Correlated Equilibria}\label{sec:cce}
\subsection{Algorithm}\label{subsec:cce_algorithm}

\begin{algorithm}[!tbp]
	\textbf{Initialize:} $Q_{h,i}^1(s,\bm{a})\gets 0,\pi_{h,i}^0(s,a_i)\gets 1/A_i, \forall s\in\mcs,h\in[H],a_i,a_i'\in\mca_i,\bm{a}\in\aall$;
	
	Set stage index $\tau\gets 1$, $t_\tau^{\text{start}}\gets 1$, and $L_\tau\gets H$\;
	
	\For{iteration $t\gets 1$ to $T$}
	{
		\textbf{Policy update:}	For all $s\in\mcs,h\in[H]$, and $a_i\in\mca_i$,
		\vspace{-.8em}\[
		\begin{aligned}
			&\ell_{h,i}^{t}(s,a_i)\gets \sum_{t'=t_{\tau}^{\text{start}}}^{t-1} [Q_{h,i}^\tau \pi_{h,-i}^{t'}](s,a_i) + [Q_{h,i}^{\tau} \pi_{h,-i}^{t-1}](s,a_i);\\
			&\pi_{h,i}^{t}(s,\cdot) \gets \argmax_{\bm{x}\in\Delta(\mca_i)}\l\langle\bm{x}, \eta_\tau \ell_{h,i}^{t}(s,\cdot)/H \rangle - \mc{R}(\bm{x})\r;
		\end{aligned}
		\]	\vspace{-1.2em}	
		
		\If{$t-t_\tau^{\text{start}} + 1\geq  L_\tau$}
		{
			$t_\tau^{\text{end}}\gets t, t_{\tau+1}^{\text{start}}\gets t+1, L_{\tau+1}\gets \floor{(1+1/H)L_\tau}$;
			
			\textbf{Value update:} For each $h\in[H], s\in\mcs, \bm{a}\in\aall,i\in\mcn$:
			\vspace{-1em}$$
			Q^{\tau+1}_{h,i}(s,\bm{a})\gets \frac{1}{L_\tau}\sum_{t'=t_\tau^{\text{start}}}^{t_\tau^{\text{end}}}\l r_{h,i} + P_h[Q^\tau_{h+1,i} \pi_{h+1}^{t'}]\r(s,\bm{a});
			$$\vspace{-1.3em}
			
			$\tau\gets\tau+1$;\	$\pi_{h,i}^t(s,a_i)\gets 1/A_i, \forall s\in\mcs,h\in[H],a_i\in\mca_i$\;
		}
	}
	\textbf{Output policy:} Sample $t\sim \operatorname{Unif}([T])$. Output $\bar{\pi}\defeq \bar{\pi}_1^t$ where $\bar{\pi}_h^t$ is defined in Algorithm~\ref{alg:certify_cce}. 
	\caption{Stage-based OFTRL for coarse correlated equilibria (agent $i$)}\label{alg:oftrl_cce}
\end{algorithm}

Algorithm~\ref{alg:oftrl_cce} describes the stage-based OFTRL procedure run by agent $i\in\mcn$ for learning CCE.
Similar to Section~\ref{sec:ce}, Algorithm~\ref{alg:oftrl_cce} also consists of three components: policy update, value update, and policy output. 
The policy update step is standard OFTRL with a negative entropy regularizer, also known as the optimistic Hedge (see e.g., \cite{chen2020hedging}). 
Our policy output step, formally described in Algorithm~\ref{alg:certify_cce} (in Appendix~\ref{app:cce}), is conceptually similar to Algorithm~\ref{alg:certify} for CE. 

The value update step here is substantially different from that of Section~\ref{sec:ce}. 
Rather than performing incremental updates as in Algorithm~\ref{alg:oftrl}, we instead employ \emph{stage-based} value updates by dividing the total $T$ iterations into multiple stages and only updating the value estimates at the end of a stage. 
We use $\tau\in\nn_+$ to index the stages and use $L_\tau$ to denote the length (i.e., number of iterations) of the $\tau$-th stage. We set the lengths of the stages to grow exponentially at a rate of $(1+1/H)$ so that $L_{\tau+1} = \floor{(1+1/H)L_\tau}$. 
The exponential growth ensures that the total $T$ iterations can be covered by a small number of stages, while the $(1+1/H)$ growth rate guarantees that the value estimation error does not blow up during the $H$ steps of recursion.
Such a mechanism was initially proposed in single-agent RL \citep{zhang2020almost} and has later been advocated for creating a piece-wise stationary environment in MARL \citep{mao2022improving}. 
The benefit of using stage-based value updates here is that we only need to bound the per-state \emph{average} regret in the corresponding matrix games (in contrast to the weighted regret as in Section~\ref{sec:ce}), which allows us to easily apply existing regret analysis results for normal-form games.

\vspace{-.5em}\subsection{Analysis}\label{subsec:cce_analysis}
We use the notion of $\operatorname{CCE-Gap}$ to measure the distance of a correlated policy to a CCE:
$
\operatorname{CCE-Gap}(\pi)\allowbreak \defeq  \max_{i\in\mcn}( V_{1,i}^{\dagger,\pi_{-i}}(s_1) - V_{1,i}^{\pi}(s_1)), 
$ 
where the best response value $V_{1,i}^{\dagger,\pi_{-i}}(s)$ is defined in Section~\ref{sec:preliminaries}. 
The following theorem shows that Algorithm~\ref{alg:oftrl_cce} finds an $\widetilde{O}(T^{-1})$-approximate CCE in $T$ iterations. 
\begin{thm}\label{thm:cce}
	If Algorithm~\ref{alg:oftrl_cce} is run on an $N$-player episodic Markov game for $T$ iterations with a learning rate $\eta_\tau = \Theta(\frac{1}{N\log^4 L_\tau})$ in each stage $\tau$, then the output policy $\bar{\pi}$ satisfies:
	\begingroup
	\setlength{\abovedisplayskip}{3pt}
	\setlength{\belowdisplayskip}{3pt}
	\setlength{\abovedisplayshortskip}{3pt}
	\setlength{\belowdisplayshortskip}{3pt}
	\[
	\operatorname{CCE-Gap}(\bar{\pi}) = O\l\frac{N H^3 \log \amax \cdot \log^5  T}{T}\r.
	\]
	\endgroup
\end{thm}
Theorem~\ref{thm:cce} improves the best-known rate of $\widetilde{O}(T^{-3/4})$  \citep{zhang2022policy} for OFTRL in general-sum Markov games. 
Compared to its counterpart $O(N\log\amax\cdot \log^4 T/T)$ in normal-form games \citep{daskalakis2021near}, Theorem~\ref{thm:cce} incurs an extra $O(\log T)$ factor due to the stage-based value estimates. 
The proof of Theorem~\ref{thm:cce} starts by showing a recursive relationship of the best response CCE value gaps between two consecutive steps $h$ and $h+1$. 
As a consequence of stage-based value updates,  $\operatorname{CCE-gap}(\bar{\pi})$ breaks down to the sum of the per-state average regret over the stages, which allows us to apply each player's individual (average) regret bound in NFGs \citep{daskalakis2021near} for each stage. 
The proof is then completed by upper bounding the total number of stages. 
We defer the complete proof of Theorem~\ref{thm:cce} to Appendix~\ref{app:cce} due to space limitations.

%% file: 5_simulations.tex
\vspace{-.5em}\section{Numerical Results}\vspace{-.2em}\label{sec:simulations}
\begin{figure}[!t]
	\centering
	\begin{minipage}{.5\textwidth}
		\centering
		\includegraphics[width=.75\linewidth]{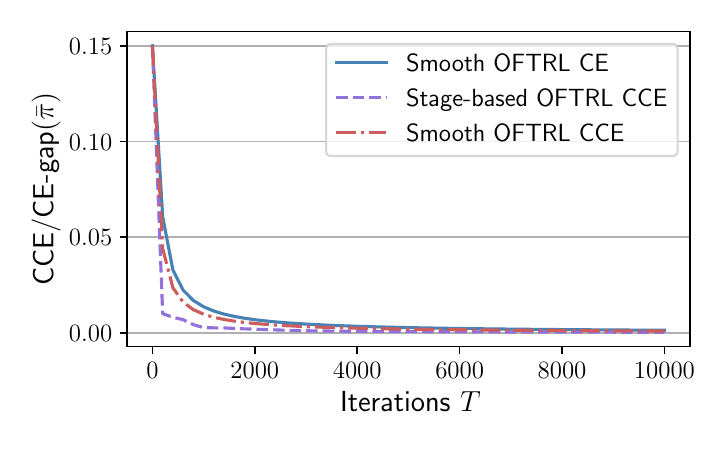}
		\captionof{figure}{\small Convergence of $\operatorname{CCE/CE-Gap}(\bar{\pi})$}
		\label{fig:1}
	\end{minipage}%
	\begin{minipage}{.5\textwidth}
		\centering
		\includegraphics[width=.75\linewidth]{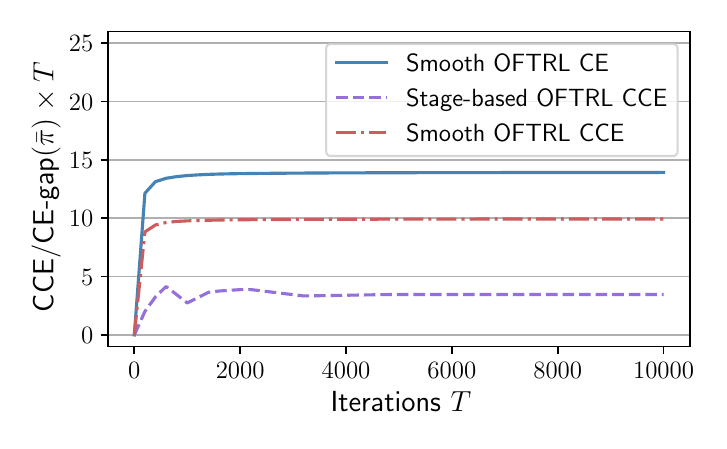}
		\captionof{figure}{\small Convergence of $\operatorname{CCE/CE-Gap}(\bar{\pi}) \times T$}
		\label{fig:2}
	\end{minipage}\vspace{-.5em}
\end{figure}
In this section, we numerically evaluate Algorithm~\ref{alg:oftrl} (denoted by ``Smooth OFTRL CE'') and Algorithm~\ref{alg:oftrl_cce} (``Stage-based OFTRL CCE'') to validate our $\widetilde{O}(T^{-1})$ theoretical convergence guarantees. 
Our simulations additionally consider an OFTRL algorithm with incremental value updates similar to that of Algorithm~\ref{alg:oftrl} for learning CCE (``Smooth OFTRL CCE''). 
We did not prove the convergence of such an algorithm but would be interested to see its numerical performance given its intuitive form. 
We conduct numerical studies on a simple general-sum Markov game with 2 players, 2 states, and 2 candidate actions for each player. 
Detailed definitions of the transition and reward functions of the game can be found in Appendix~\ref{app:simulations}. 
Figure~\ref{fig:1} illustrates the convergence of the three algorithms to their corresponding equilibrium solutions as the number of iterations increases. 
To clearly demonstrate their convergence rates, we further plot the behavior of $\operatorname{CCE/CE-Gap}(\bar{\pi}) \times T$ as $T$ increases.
We can observe from Figure~\ref{fig:2} that for all three algorithms, $\operatorname{CCE/CE-Gap}(\bar{\pi}) \times T$ essentially become a constant for any reasonably large value of $T$. 
This indicates that our algorithms indeed converge at a rate of $\widetilde{O}(T^{-1})$ numerically. 
We also observe that OFTRL with stage-based value updates numerically converges faster than its incrementally-updated counterpart despite using the same learning rate, which advocates the use of stage-based value updates in Markov games.

%% file: appendix_lemma.tex
\section{Technical Lemmas}\label{app:lemmas}

\begin{lem}\label{lemma:matrix_swap_regret}
	(Extension of Theorem 4.3 in \cite{anagnostides2022uncoupled} to time-varying learning rates) In a no-regret learning problem as defined in Section~\ref{sec:preliminaries}, suppose that BM-OFTRL \eqref{eqn:oftrl} is run with log-barrier regularization and a time-varying learning rate $\eta_t \leq \frac{1}{128\sqrt{\abs{\mca}}},\forall t\in[T]$. 
	Then, for any $T\geq 2$, the swap regret is bounded by
	\[
	\operatorname{SwapReg}^T \leq \frac{2\abs{\mca}^2\log T}{\eta_T} + 4\sum_{t=1}^T \eta_t \norm{\bm{u}^t - \bm{u}^{t-1}}_\infty^2 - \frac{1}{2048 \abs{\mca}} \sum_{t=1}^{T-1} \frac{1}{\eta_t}\norm{\bm{x}^{t+1}-\bm{x}^t}_1^2.
	\]
\end{lem}
\begin{proofsketch}
	The proof follows a similar procedure as that of Theorem 4.3 in \cite{anagnostides2022uncoupled}, except that we need to re-derive their Theorems B.1 and 3.1 under a time-varying learning rate. 
	We skip the proof here as such an extension is straightforward. 
\end{proofsketch}

\begin{lem}\label{lemma:matrix_external_regret}
	(Theorem 3.1 from \cite{daskalakis2021near}) In a normal-form game with $N$ players and $A_i$ actions for player $i\in[N]$, suppose that all the players run OFTRL for $T$ steps with negative entropy
	regularization and a learning rate $\eta = \Theta(\frac{1}{N\log^4 T})$.
	Then, there exists a constant $C>1$ such that the regret of player $i$ satisfies
	\[
		\operatorname{Reg}_i^T \leq C N \log A_i \cdot \log^4 T.
	\]
\end{lem}

%% file: appendix_ce.tex
\section{Proofs for Section~\ref{sec:ce}}\label{app:ce}
\noindent\textbf{Lemma \ref{lemma:recursion}.} 	(Recursion of best response CE value gaps) For any fixed $(h,t)\in[H]\times[T]$, we have
	\[
	\delta_h^t \leq \sum_{j=1}^t \alpha_t^j \delta_{h+1}^j + \operatorname{SwapReg}_h^t.
	\]
\begin{proof}
	For any fixed $i\in\mcn$ and $s\in\mcs$, we know from the definition of $\bar{\pi}_h^t$ from Algorithm~\ref{alg:certify} that
	\begin{equation}\label{eqn:a9}
		V_{h,i}^{\bar{\pi}_h^t}(s) = \sum_{j=1}^t \alpha_t^j \inner{\pi_{h,i}^j(s,\cdot), \L \l r_{h,i} + [P_hV_{h+1,i}^{\bar{\pi}_{h+1}^j}] \r \pi_{h,-i}^j \R(s,\cdot)}.
	\end{equation}		
	For a fixed $\bar{\pi}_h^t$, we use $\phi_i^\star$ to denote the best response strategy modification that maximizes the value function starting from step $h$. 
	In this case, we know from the definition of the value function that 
	\[
	\begin{aligned}
		& V_{h,i}^{\phi_i^\star \diamond \bar{\pi}_h^t}(s) =  \max_{\phi_{h,i}^s:\mca_i\ra\mca_i}\sum_{j=1}^t \alpha_t^j \inner{\phi_{h,i}^s\diamond \pi_{h,i}^j(s,\cdot), \L \l r_{h,i} + [P_hV_{h+1,i}^{\phi_i^\star\diamond \bar{\pi}_{h+1}^j}] \r \pi_{h,-i}^j \R(s,\cdot)}\\
		= & \max_{\phi_{h,i}^s}\sum_{j=1}^t \alpha_t^j \inner{\phi_{h,i}^s\diamond \pi_{h,i}^j(s,\cdot), \L \l r_{h,i} + [P_hV_{h+1,i}^{\bar{\pi}_{h+1}^j}] + [P_hV_{h+1,i}^{\phi_i^\star\diamond \bar{\pi}_{h+1}^j}] - [P_hV_{h+1,i}^{\bar{\pi}_{h+1}^j}] \r \pi_{h,-i}^j \R(s,\cdot)}\\
		\leq & \max_{\phi_{h,i}^s}\sum_{j=1}^t \alpha_t^j \Big\langle\phi_{h,i}^s\diamond \pi_{h,i}^j(s,\cdot), \Big[ \big( r_{h,i} + [P_hV_{h+1,i}^{\bar{\pi}_{h+1}^j}] \big) \pi_{h,-i}^j \Big](s,\cdot)\Big\rangle 
		+ \sum_{j=1}^t \alpha_t^j  \max_{s'\in\mcs}\Big( V_{h+1,i}^{\phi_i^\star\diamond \bar{\pi}_{h+1}^j} 
		- V_{h+1,i}^{\bar{\pi}_{h+1}^j} \Big)(s'). 
	\end{aligned}
	\]
	Subtracting \eqref{eqn:a9} from the above equation leads to:
	\begin{align}
		V_{h,i}^{\phi_i^\star \diamond \bar{\pi}_h^t}(s) -& V_{h,i}^{\bar{\pi}_h^t}(s) 
		\leq \sum_{j=1}^t \alpha_t^j  
		\max_{s'\in\mcs}\Big( V_{h+1,i}^{\phi_i^\star\diamond \bar{\pi}_{h+1}^j} (s')
		- V_{h+1,i}^{\bar{\pi}_{h+1}^j}(s') \Big) \nonumber\\
		& + \max_{\phi_{h,i}^s}\sum_{j=1}^t \alpha_t^j \Big\langle\phi_{h,i}^s\diamond \pi_{h,i}^j(s,\cdot) -  \pi_{h,i}^j(s,\cdot), \Big[ \big( r_{h,i} + [P_hV_{h+1,i}^{\bar{\pi}_{h+1}^j}] \big) \pi_{h,-i}^j \Big](s,\cdot)\Big\rangle. \label{eqn:a8}
	\end{align}
	In the following, we will show that \eqref{eqn:a8} is equal to $\operatorname{SwapReg}_{h,i}^t(s)$. 
	It suffices to show that $Q_{h,i}^t(s, \bm{a}) = \big(r_{h,i} + [P_hV_{h+1,i}^{\bar{\pi}_{h+1}^t}]\big)(s, \bm{a}), \forall t\in[T], \bm{a}\in\aall$. 
	We prove this claim by backward induction over $h\in[H]$. 
	Notice that the claim trivially holds for $h=H$ as $Q_{H,i}^t(s, \bm{a}) = r_{H,i}(s, \bm{a}),  \forall t\in[T], \bm{a}\in\aall$. 
	Suppose that the claim holds for $h$; then, for step $h-1$, we have that
	\[
	\begin{aligned}
		Q_{h-1,i}^t(s, \bm{a}) = &\sum_{j=1}^t \alpha_t^j \l r_{h-1,i} + P_{h-1}[Q_{h,i}^j\pi_{h}^j]\r(s,\bm{a})\\
		=&  r_{h-1,i}(s,\bm{a}) + P_{h-1}\Big[\sum_{j=1}^t \alpha_t^j Q_{h,i}^j\pi_{h}^j\Big](s,\bm{a})\\
		=& r_{h-1,i}(s,\bm{a}) + P_{h-1}\Big[\sum_{j=1}^t \alpha_t^j \Big(r_{h,i} + [P_{h}V_{h+1,i}^{\bar{\pi}_{h+1}^j}]\Big)\pi_{h}^j\Big](s,\bm{a})\\
		=& r_{h-1,i}(s,\bm{a}) + \Big[P_{h-1} V_{h,i}^{\bar{\pi}_h^t}\Big](s,\bm{a}), 
	\end{aligned}
	\]
	where the first step is by \eqref{eqn:value_updates}, the second step changes the order of summation, the third step uses the induction hypothesis, and the last step is due to \eqref{eqn:a9}. 
	This completes the proof of $Q_{h,i}^t(s, \bm{a}) = \big(r_{h,i} + [P_hV_{h+1,i}^{\bar{\pi}_{h+1}^t}]\big)(s, \bm{a})$.
	Substituting it back to \eqref{eqn:a8}, we obtain that 
	\[
	V_{h,i}^{\phi_i^\star \diamond \bar{\pi}_h^t}(s) - V_{h,i}^{\bar{\pi}_h^t}(s) 
	\leq \sum_{j=1}^t \alpha_t^j  
	\max_{s'\in\mcs}\Big( V_{h+1,i}^{\phi_i^\star\diamond \bar{\pi}_{h+1}^j} (s')
	- V_{h+1,i}^{\bar{\pi}_{h+1}^j}(s') \Big) + \operatorname{SwapReg}_{h,i}^t(s). 
	\]
	Since the above inequality holds for any $i\in\mcn$ and $s\in\mcs$, and since  $V_{h+1,i}^{\phi_i^\star\diamond \bar{\pi}_{h+1}^j} (s') \leq \max_{\phi_i}V_{h+1,i}^{\phi_i\diamond \bar{\pi}_{h+1}^j} (s') $ at step $h+1$, we can conclude that 
	\[
	\delta_h^t \leq \sum_{j=1}^t \alpha_t^j \delta_{h+1}^j + \operatorname{SwapReg}_h^t,
	\]
	This completes the proof of the recursive relationship of best response CE value gaps. 
\end{proof}

\noindent\textbf{Lemma \ref{lemma:perstate_weighted_regret}.} (Per-state weighted swap regret bounds) For any $t\in[T], h\in[H], s\in\mcs$ and $i\in\mcn$, Algorithm~\ref{alg:oftrl} ensures that
	\begin{align}
		\operatorname{SwapReg}_{h,i}^t(s) \leq& \frac{4A_i^2 H\log t}{\eta t}  +  \frac{32\eta H^3N^2}{t} + 8\eta N H^2\sum_{j=2}^t\sum_{k\in\mcn, k\neq i}\alpha_t^j \norm{\pi_{h,k}^j(s,\cdot) - \pi_{h,k}^{j-1}(s,\cdot)}_1^2\nonumber\\
		 &- \frac{1}{2048\eta A_i}\sum_{j=2}^t \alpha_t^{j-1}\norm{\pi_{h,i}^j(s,\cdot) - \pi_{h,i}^{j-1}(s,\cdot)}_1^2.\nonumber
	\end{align}
\noindent Consequently, if $\eta \leq \frac{1}{256NH \sqrt{H\amax}}$, we further have
	\[
	\begin{aligned}
		\sum_{i=1}^N\operatorname{SwapReg}_{h,i}^t(s) \leq &\frac{4 N\amax^2 H\log t}{\eta t}  + \frac{32\eta NH^2(N^2+H)}{t}\\
		&- \frac{1}{2048\eta H }\sum_{i=1}^N\sum_{j=2}^t  \frac{\alpha_t^j}{A_i}\norm{\pi_{h,i}^{j}(s,\cdot) - \pi_{h,i}^{j-1}(s,\cdot)}_1^2.
	\end{aligned}
	\]
\begin{proof}
	At each fixed $(s,h)\in\mcs\times[H]$, the agents essentially face a no-swap-regret learning problem in a matrix game, where the payoff matrix of agent $i$ is $Q_{h,i}^t(s,\cdot)$ at iteration $t$. 
	We can apply the weighted swap regret bound (Lemma~\ref{lemma:matrix_swap_regret} in Appendix~\ref{app:lemmas}) of OFTRL under the Blum-Mansour reduction in normal-form games to obtain:
	\begin{align}
		\operatorname{SwapReg}_{h,i}^t(s) 
		=& \max_{\phi_{h,i}^s:\mca_i\ra\mca_i}\sum_{j=1}^t \alpha_t^j\inner{\phi_{h,i}^s\diamond \pi_{h,i}^j(s,\cdot) - \pi_{h,i}^j(s,\cdot), [Q_{h,i}^j \pi_{h,-i}^j](s, \cdot)}\nonumber\\
		=& \alpha_t^1 \max_{\phi_{h,i}^s:\mca_i\ra\mca_i} \sum_{j=1}^t \inner{\phi_{h,i}^s\diamond \pi_{h,i}^j(s,\cdot) - \pi_{h,i}^j(s,\cdot), w_j[Q_{h,i}^j \pi_{h,-i}^j](s, \cdot)}\label{eqn:a7}\\
		\leq & \frac{2A_i^2\alpha_t \log t}{\eta} + 4\sum_{j=1}^t \frac{\eta \alpha_t^1}{w_j}\norm{w_j[Q_{h,i}^j \pi_{h,-i}^j](s, \cdot) - w_{j}[Q_{h,i}^{j-1} \pi_{h,-i}^{j-1}](s, \cdot)}_\infty^2 \nonumber\\
		& - \frac{\alpha_t^1}{2048\eta A_i}\sum_{j=2}^t w_{j-1}\norm{\pi_{h,i}^j(s,\cdot) - \pi_{h,i}^{j-1}(s,\cdot)}_1^2,\label{eqn:a6}
	\end{align}
	where \eqref{eqn:a7} is due to the choice of the weights $w_j = \alpha_t^j/\alpha_t^1$.
	\eqref{eqn:a6} uses Lemma~\ref{lemma:matrix_swap_regret} from Appendix~\ref{app:lemmas}, by instantiating $\bm{u}^j(\cdot)$ in Lemma~\ref{lemma:matrix_swap_regret} as $w_j[Q_{h,i}^j\pi_{h,-i}^j](s, \cdot)$, the prediction $\bm{m}^t=w_j[Q_{h,i}^{j-1}\pi^{j-1}_{h,-i}](s, \cdot)$, and the learning rate $\eta_j=\eta/w_j$. 
	To further upper bound the above equation, notice that
	\begin{align}
		&\sum_{j=1}^t \frac{\eta \alpha_t^1}{w_j}\norm{w_j[Q_{h,i}^j \pi_{h,-i}^j](s, \cdot) - w_{j}[Q_{h,i}^{j-1} \pi_{h,-i}^{j-1}](s, \cdot)}_\infty^2\nonumber\\
		=& \sum_{j=1}^t \eta \alpha_t^1w_j\norm{\l [Q_{h,i}^j \pi_{h,-i}^j] - [Q_{h,i}^{j-1} \pi_{h,-i}^j] + [Q_{h,i}^{j-1} \pi_{h,-i}^j] -[Q_{h,i}^{j-1} \pi_{h,-i}^{j-1}]\r (s, \cdot)}_\infty^2\nonumber\\
		\leq & 2\sum_{j=1}^t \eta \alpha_t^1w_j\l \norm{Q_{h,i}^{j}(s,\cdot) - Q_{h,i}^{j-1}(s,\cdot)}_\infty^2 + H^2\norm{\pi_{h,-i}^{j}(s,\cdot) - \pi_{h,-i}^{j-1}(s,\cdot)}_1^2 \r\nonumber\\
		\leq & 2\sum_{j=1}^t \eta \alpha_t^1w_j (\alpha_j)^2H^2 + 2\sum_{j=1}^t \eta \alpha_t^1w_j H^2\norm{\pi_{h,-i}^{j}(s,\cdot) - \pi_{h,-i}^{j-1}(s,\cdot)}_1^2,  \label{eqn:a5}
	\end{align}
	where the second step uses the observation that $(a+b)^2\leq 2a^2+2b^2$, the H\"{o}lder's inequality, and the fact that $\|Q_{h,i}^{j-1}\|_\infty\leq H$. 
	The third step is due to our value update rule in Algorithm~\ref{alg:oftrl}, which yields
	\[
	\begin{aligned}
		\norm{Q_{h,i}^{j}(s,\cdot) - Q_{h,i}^{j-1}(s,\cdot)}_\infty = & 
		\norm{-\alpha_j Q_{h,i}^{j-1} (s,\cdot) + \alpha_j \l r_{h,i} + P_h[Q_{h+1,i}^j\pi_{h+1}^j]\r(s,\cdot)}_\infty\\
		\leq & \alpha_j \max\left\{\norm{Q_{h,i}^{j-1} (s,\cdot)}_\infty, \norm{\l r_{h,i} + P_h[Q_{h+1,i}^j\pi_{h+1}^j]\r(s,\cdot)}_\infty\right\}\\
		\leq & \alpha_j H.
	\end{aligned}
	\]
	To continue from \eqref{eqn:a5}, we apply the properties that $w_j = \alpha_t^j/\alpha_t^1$ and $\sum_{j=1}^t\alpha_t^j (\alpha_j)^2\leq \sum_{j=1}^t (\alpha_j)^2/t \leq (H+2)/t \leq 3H/t$ (see Lemma 6 in \cite{yang2022t} for a proof) to obtain:
	\begin{align}
		\eqref{eqn:a5} =& 2\sum_{j=1}^t \eta \alpha_t^1w_j (\alpha_j)^2H^2 + 2\sum_{j=1}^t \eta \alpha_t^1w_j H^2\norm{\pi_{h,-i}^{j}(s,\cdot) - \pi_{h,-i}^{j-1}(s,\cdot)}_1^2\nonumber\\
		\leq &\frac{6\eta H^3}{t} + 2\sum_{j=1}^t \eta \alpha_t^1w_j H^2\norm{\pi_{h,-i}^{j}(s,\cdot) - \pi_{h,-i}^{j-1}(s,\cdot)}_1^2\nonumber\\
		\leq & \frac{6\eta H^3}{t} +2\eta  (N-1)H^2\sum_{j=1}^t  \alpha_t^j \sum_{k\in\mcn, k\neq i}\norm{\pi_{h,k}^{j}(s,\cdot) - \pi_{h,k}^{j-1}(s,\cdot)}_1^2.\label{eqn:a4}
	\end{align}
	In the last step, we used that the total variation between two product distributions is bounded by the sum of the total variations of each marginal distribution (see e.g.~\citep{hoeffding1958distinguishability}):
	\[
	\begin{aligned}
		\norm{\pi_{h,-i}^{j}(s,\cdot) - \pi_{h,-i}^{j-1}(s,\cdot)}_1^2 = &\bigg( \sum_{\bm{a}_{-i}\in\mca_{-i}} \abs{\pi_{h,-i}^{j}(s,\bm{a}_{-i}) - \pi_{h,-i}^{j-1}(s,\bm{a}_{-i})} \bigg)^2\\
		=& \bigg( \sum_{\bm{a}_{-i}\in\mca_{-i}} \bigg|\prod_{k\neq i}\pi_{h,k}^{j}(s, a_{k}) - \prod_{k\neq i} \pi_{h,k}^{j-1}(s, a_{k})\bigg| \bigg)^2\\
		\leq & \bigg( \sum_{k\neq i}\norm{\pi_{h,k}^j(s, a_k) - \pi_{h,k}^{j-1}(s, a_k)}_1 \bigg)^2\\
		\leq & (N-1) \sum_{k\neq i}\norm{\pi_{h,k}^j(s, a_k) - \pi_{h,k}^{j-1}(s, a_k)}_1^2,
	\end{aligned}
	\]
	and the last step is by the Cauchy–Schwarz inequality. 
	Substituting \eqref{eqn:a4} back to \eqref{eqn:a6} leads to
	\begin{align}
		\operatorname{SwapReg}_{h,i}^t(s) \leq& \frac{2A_i^2\alpha_t \log t}{\eta} +8\eta  (N-1)H^2\sum_{j=1}^t  \alpha_t^j \sum_{k\in\mcn, k\neq i}\norm{\pi_{h,k}^{j}(s,\cdot) - \pi_{h,k}^{j-1}(s,\cdot)}_1^2 \nonumber\\
		&+ \frac{24\eta H^3}{t}  - \frac{\alpha_t^1}{2048\eta A_i}\sum_{j=2}^t w_{j-1}\norm{\pi_{h,i}^j(s,\cdot) - \pi_{h,i}^{j-1}(s,\cdot)}_1^2 \nonumber\\
		\leq & \frac{4A_i^2 H\log t}{\eta t}  +8\eta  (N-1)H^2\sum_{j=2}^t  \alpha_t^j \sum_{k\in\mcn, k\neq i}\norm{\pi_{h,k}^{j}(s,\cdot) - \pi_{h,k}^{j-1}(s,\cdot)}_1^2 \nonumber\\
		&+ \frac{32\eta H^2(H+N^2)}{t}  - \frac{1}{2048\eta A_i}\sum_{j=2}^t \alpha_t^{j-1}\norm{\pi_{h,i}^j(s,\cdot) - \pi_{h,i}^{j-1}(s,\cdot)}_1^2,	\label{eqn:a3}
	\end{align}
	where the second inequality uses $\alpha_t = (H+1)/(H+t)\leq 2H/t$. This step also takes out the term for $j=1$ and upper bounds it by
	\[
	8\eta  (N-1)H^2 \alpha_t^1 \sum_{k\in\mcn, k\neq i}\norm{\pi_{h,k}^{1}(s,\cdot) - \pi_{h,k}^{0}(s,\cdot)}_1^2 \leq \frac{32\eta (N-1)^2 H^2}{t}, 
	\]
	using the fact that $\alpha_t^1 \leq 1/t$ (Lemma 6 in \cite{yang2022t}). 
	This proves the first claim in the lemma. 
	To further establish the second statement, we sum over \eqref{eqn:a3} to obtain
	\[
	\begin{aligned}
		\sum_{i=1}^N\operatorname{SwapReg}_{h,i}^t(s) \leq & \frac{4 NA_i^2 H\log t}{\eta t}  +8\eta  (N-1)^2H^2\sum_{i=1}^N\sum_{j=2}^t    \alpha_t^j \norm{\pi_{h,i}^{j}(s,\cdot) - \pi_{h,i}^{j-1}(s,\cdot)}_1^2 \\
		&+ \frac{32\eta NH^2(H+N^2)}{t}  - \frac{1}{2048\eta }\sum_{i=1}^N\sum_{j=2}^t \frac{\alpha_t^{j-1}}{A_i}\norm{\pi_{h,i}^j(s,\cdot) - \pi_{h,i}^{j-1}(s,\cdot)}_1^2\\
		\leq & \frac{4 NA_i^2 H\log t}{\eta t}  + \frac{32\eta NH^2(H+N^2)}{t}\\
		& + \sum_{i=1}^N\sum_{j=2}^t   \l 8\eta  (N-1)^2H^2  - \frac{1}{2048\eta H A_i }\r\alpha_t^j\norm{\pi_{h,i}^{j}(s,\cdot) - \pi_{h,i}^{j-1}(s,\cdot)}_1^2 \\
		\leq & \frac{4 NA_i^2 H\log t}{\eta t}+ \frac{32\eta NH^2(H+N^2)}{t}\\
		&- \frac{1}{2048\eta H }\sum_{i=1}^N\sum_{j=2}^t  \frac{\alpha_t^j}{A_i}\norm{\pi_{h,i}^{j}(s,\cdot) - \pi_{h,i}^{j-1}(s,\cdot)}_1^2,
	\end{aligned}
	\]
	where the second step uses the fact that $\alpha_t^{j-1}/\alpha_t^j = (j-1)/(H+j-1) \geq 1/H$, and the last step is due to the condition that $\eta \leq\frac{1}{256NH \sqrt{H\amax}}$. 
\end{proof}

\noindent\textbf{Theorem \ref{thm:main}.}
	If Algorithm~\ref{alg:oftrl} is run on an $N$-player episodic Markov game for $T$ iterations with a learning rate $\eta = \frac{1}{256 NH \sqrt{H\amax}}$, the output policy $\bar{\pi}$ satisfies:
	\[
	\operatorname{CE-Gap}(\bar{\pi}) \leq \frac{6144 N H^{\frac{7}{2}} \amax^{\frac{5}{2}}\log T}{T}. 
	\]
\begin{proof}
	Using \eqref{eqn:a0} from Lemma~\ref{lemma:perstate_weighted_regret}, we upper bound the second-order path lengths by
	\[
	\begin{aligned}
		&8\eta NH^2\sum_{i=1}^N\sum_{j=2}^t  \alpha_t^j\norm{\pi_{h,i}^{j}(s,\cdot) - \pi_{h,i}^{j-1}(s,\cdot)}_1^2 \\
		\leq &8\eta NH^2\cdot 2048\eta H \amax\l \frac{4 N\amax^2 H\log t}{\eta t}  + \frac{32\eta NH^2(N^2+H)}{t}\r,
	\end{aligned}
	\]
	where we used the crucial fact that the swap regret is non-negative. 
	Substituting the above equation back to \eqref{eqn:a1} yields
	\begin{align}
		\operatorname{SwapReg}_{h,i}^t(s)  \leq &\frac{4A_i^2 H\log t}{\eta t}  +  \frac{32\eta H^3N^2}{t} + \frac{2^{16}\eta N^2 H^4 \amax^3\log t}{t} + \frac{2^{19}\eta^3 N^4 H^6}{t}\nonumber\\
		\leq & \frac{2048N H^{\frac{5}{2}} \amax^{\frac{5}{2}} \log t}{t},\label{eqn:b8}
	\end{align}
	where the second step uses $\eta = \frac{1}{256NH \sqrt{H\amax}}$. 
	Since \eqref{eqn:b8} holds for any $i\in\mcn$ and $s\in\mcs$, we can apply it back to the recursion of best response CE value gaps from Lemma~\ref{lemma:recursion} to obtain
	\[
	\delta_h^t \leq \sum_{j=1}^t \alpha_t^j \delta_{h+1}^j + \frac{2048N H^{\frac{5}{2}} \amax^{\frac{5}{2}} \log t}{t}.
	\]
	Starting from $\delta_{H+1}^t = 0$, we can show via backward induction that for any $(h,t)\in[H]\times[T]$, 
	\[
	\delta_h^t \leq \frac{6144 N \amax^{\frac{5}{2}}(H-h+1)H^{\frac{5}{2}}\log t}{t},
	\]
	where we applied Lemma 4 from \cite{yang2022t} that $\sum_{j=1}^t \alpha_t^j/j\leq (1+\frac{1}{H})\frac{1}{t}$. 
	We conclude the proof of the theorem by referring to the property that $\operatorname{CE-Gap}(\bar{\pi})\leq \delta_1^T$. 
\end{proof}

%% file: appendix_cce.tex
\section{Supplementary Material for Section~\ref{sec:cce}}\label{app:cce}
\subsection{Policy Output Algorithm}

\begin{algorithm}[!h]
	\textbf{Input:} Policy trajectory $\{\pi_h^t\}_{h\in[H],t\in[T]}$ of Algorithm~\ref{alg:oftrl_cce};
	
	\For{step $h'\gets h$ to $H$}
	{
		Uniformly sample $j$ from $\{t_{\tau(t)-1}^{\text{start}}, t_{\tau(t)-1}^{\text{start}}+1,\dots, t_{\tau(t)-1}^{\text{end}}\}$;
		
		Play policy $\pi_{h'}^j$ for step $h'$;
		
		Set $t\gets j$.
	}
	\caption{Policy $\bar{\pi}_h^t$ for stage-based OFTRL}\label{alg:certify_cce}
\end{algorithm}

\subsection{Proofs}
\begin{lem}\label{lemma:recursion_cce}
	(Recursion of best response CCE value gaps) For any fixed $(h,t)\in[H]\times[T]$, let $\tau = \tau(t)$ denote the stage of $t$. 
	Then, we have
	\[
		\zeta^t_h \leq \frac{1}{L_{\tau-1}} \sum_{j=t_{\tau-1}^{\text{start}}}^{t_{\tau-1}^{\text{end}}} \zeta_{h+1}^j + \operatorname{Reg}_{h}^{\tau-1}.
	\]
\end{lem}
\begin{proof}
	For any fixed $i\in\mcn$ and $s\in\mcs$, we know from the definition of $\bar{\pi}_h^t$ from Algorithm~\ref{alg:certify_cce} that
	\begin{equation}\label{eqn:d9}
		V_{h,i}^{\bar{\pi}_h^t}(s) = \frac{1}{L_{\tau-1}} \sum_{j=t_{\tau-1}^{\text{start}}}^{t_{\tau-1}^{\text{end}}} \inner{\pi_{h,i}^j(s,\cdot), \L \l r_{h,i} + [P_hV_{h+1,i}^{\bar{\pi}_{h+1}^j}] \r \pi_{h,-i}^j \R(s,\cdot)}.
	\end{equation}		
	From the definition of the best response value function,
	\[
	\begin{aligned}
		 V_{h,i}^{\dagger, \bar{\pi}_{h,-i}^t}(s) = & \max_{\pi_i^\dagger(s,\cdot)\in\Delta(\mca_i)}\frac{1}{L_{\tau-1}} \sum_{j=t_{\tau-1}^{\text{start}}}^{t_{\tau-1}^{\text{end}}} \Big\langle\pi_i^\dagger(s,\cdot), \big[ \big( r_{h,i} + [P_hV_{h+1,i}^{\dagger,\bar{\pi}_{h+1,-i}^j}] \big) \pi_{h,-i}^j \big](s,\cdot)\Big\rangle\\
		= & \max_{\pi_i^\dagger}\frac{1}{L_{\tau-1}} \sum_{j=t_{\tau-1}^{\text{start}}}^{t_{\tau-1}^{\text{end}}} \Big\langle\pi_i^\dagger(s,\cdot), \big[ \big( r_{h,i} + [P_hV_{h+1,i}^{\bar{\pi}_{h+1}^j}] 
		- [P_hV_{h+1,i}^{\bar{\pi}_{h+1}^j}] +  [P_hV_{h+1,i}^{\dagger,\bar{\pi}_{h+1,-i}^j}] \big) \pi_{h,-i}^j \big](s,\cdot)\Big\rangle\\
		\leq & \max_{\pi_i^\dagger}\frac{1}{L_{\tau-1}} \sum_{j=t_{\tau-1}^{\text{start}}}^{t_{\tau-1}^{\text{end}}} \Big\langle\pi_i^\dagger(s,\cdot), \big[ \big( r_{h,i} + [P_hV_{h+1,i}^{\bar{\pi}_{h+1}^j}] \big) \pi_{h,-i}^j \big](s,\cdot)\Big\rangle\\
		&+ \frac{1}{L_{\tau-1}} \sum_{j=t_{\tau-1}^{\text{start}}}^{t_{\tau-1}^{\text{end}}} \max_{s'\in\mcs}\Big( V_{h+1,i}^{\dagger, \bar{\pi}_{h+1,-i}^j} 
		- V_{h+1,i}^{\bar{\pi}_{h+1}^j} \Big)(s'). 
	\end{aligned}
	\]
	Subtracting \eqref{eqn:d9} from the above equation leads to:
	\begin{align}
	V_{h,i}^{\dagger, \bar{\pi}_{h,-i}^t}(s) -& V_{h,i}^{\bar{\pi}_h^t}(s) 
		\leq \frac{1}{L_{\tau-1}} \sum_{j=t_{\tau-1}^{\text{start}}}^{t_{\tau-1}^{\text{end}}} \max_{s'\in\mcs}\Big( V_{h+1,i}^{\dagger, \bar{\pi}_{h+1,-i}^j} 
		- V_{h+1,i}^{\bar{\pi}_{h+1}^j} \Big)(s') \nonumber\\
		& + \max_{\pi_i^\dagger}\frac{1}{L_{\tau-1}} \sum_{j=t_{\tau-1}^{\text{start}}}^{t_{\tau-1}^{\text{end}}} \Big\langle\pi_i^\dagger(s,\cdot) - \pi_{h,i}^j(s,\cdot), \big[ \big( r_{h,i} + [P_hV_{h+1,i}^{\bar{\pi}_{h+1}^j}] \big) \pi_{h,-i}^j \big](s,\cdot)\Big\rangle. \label{eqn:d8}
	\end{align}
	Using a similar inductive argument as in the proof of Lemma~\ref{lemma:recursion}, we can show that the term in \eqref{eqn:d8} is equal to $\operatorname{Reg}_{h,i}^{\tau-1}(s)$, which leads to 
	\[
	V_{h,i}^{\dagger, \bar{\pi}_{h,-i}^t}(s) - V_{h,i}^{\bar{\pi}_h^t}(s) 
	\leq \frac{1}{L_{\tau-1}} \sum_{j=t_{\tau-1}^{\text{start}}}^{t_{\tau-1}^{\text{end}}} \max_{s'\in\mcs}\Big( V_{h+1,i}^{\dagger, \bar{\pi}_{h+1,-i}^j} 
	- V_{h+1,i}^{\bar{\pi}_{h+1}^j} \Big)(s') + \operatorname{Reg}_{h,i}^{\tau-1}(s). 
	\]
	Since the above inequality holds for any $i\in\mcn$ and $s\in\mcs$, we can conclude that 
	\[
	\zeta^t_h \leq \frac{1}{L_{\tau-1}} \sum_{j=t_{\tau-1}^{\text{start}}}^{t_{\tau-1}^{\text{end}}} \zeta_{h+1}^j + \operatorname{Reg}_{h}^{\tau-1}.
	\]
	This completes the proof of the recursive relationship of best response CCE value gaps. 
\end{proof}

\noindent\textbf{Theorem~\ref{thm:cce}.} If Algorithm~\ref{alg:oftrl_cce} is run on an $N$-player episodic Markov game for $T$ iterations with a learning rate $\eta_\tau = \Theta(\frac{1}{N\log^4 L_\tau})$ in each stage $\tau$, the output policy $\bar{\pi}$ satisfies:
\[
\operatorname{CCE-Gap}(\bar{\pi}) = O\l\frac{N H^3 \log \amax \cdot \log^5  T}{T}\r.
\]
\begin{proof} 
We introduce a few more notations before presenting the proof.
Let $\tau(t)$ denote the index of the stage that iteration $t$ belongs to. 
We denote by $\bar{\tau}$ the total number of stages, i.e., $\bar{\tau}\defeq \tau(T)$. 
For any $(\tau,h,s)$, we define the per-state (average) regret for player $i\in\mcn$ in the $\tau$-th stage of the corresponding matrix game as
\[
\begin{aligned}
	\operatorname{Reg}_{h, i}^\tau(s) &\defeq\max _{\pi_{i}^{\dagger}(s, \cdot) \in \Delta(\mca_i)} \frac{1}{L_\tau} \sum_{j=t_\tau^{\text{start}}}^{t_\tau^{\text{end}}} \left\langle \pi_{i}^{\dagger}(s,\cdot)-\pi_{h,i}^j(s,\cdot), [Q_{h,i}^\tau \pi_{h,-i}^j](s,\cdot)\right\rangle,\\
	\operatorname{Reg}_{h}^\tau &\defeq \max_{i\in\mcn}\max_{s\in\mcs}\operatorname{Reg}_{h, i}^\tau(s),
\end{aligned}
\]
where $Q_{h,i}^\tau$ is player $i$'s Q-function estimate at stage $\tau$. 
For any $(h,t)\in[H]\times[T]$ and for the policy $\bar{\pi}_h^t$ as defined in Algorithm~\ref{alg:certify_cce}, we define the best response CCE value gap as
\[
\zeta_h^t \defeq \max_{i\in\mcn}\max_{s\in\mcs} \l V_{h,i}^{\dagger, \bar{\pi}_{h,-i}^t}(s) - V_{h,i}^{\bar{\pi}_h^t}(s)\r.
\]
	By the definition of $\bar{\pi}$ and $\zeta_h^t$, we have
	\begin{align}
		\operatorname{CCE-gap}(\bar{\pi}) = & \max_{i\in\mcn}\l V_{1,i}^{\dagger,\bar{\pi}_{-i}}(s_1)- V_{1,i}^{\bar{\pi}}(s_1)
		\r\nonumber \\
		\leq &\frac{1}{T}\sum_{t=1}^T \max_{i\in\mcn}\max_{s\in\mcs} \l V_{1,i}^{\dagger,\bar{\pi}^t_{1,-i}}(s)- V_{1,i}^{\bar{\pi}_1^t}(s) \r
		\leq \frac{1}{T}\sum_{t=1}^T \zeta_1^t.\label{eqn:c6}
	\end{align}
	We use Lemma~\ref{lemma:recursion_cce}  to establish the following recursive relationship of the best response CCE value gaps between two consecutive steps $h$ and $h+1$:
	\begin{equation}\label{eqn:c9}
		\zeta^t_h \leq \frac{1}{L_{\tau-1}} \sum_{j=t_{\tau-1}^{\text{start}}}^{t_{\tau-1}^{\text{end}}} \zeta_{h+1}^j + \operatorname{Reg}_{h}^{\tau-1}.
	\end{equation}
	Hence, upper bounding $\operatorname{CCE-gap}(\bar{\pi})$ breaks down to controlling the per-state regret in the corresponding matrix game for each $(\tau, s,h)\in[\bar{\tau}]\times \mcs\times[H]$. 
	In our stage-based OFTRL, since the reward matrix $Q_{h,i}^\tau$ in each stage is fixed and $\operatorname{Reg}_{h, i}^\tau(s)$ is the standard (average) regret, we can readily apply the individual regret bound of each player when running OFTRL in normal-form games \citep{daskalakis2021near}. 
	Specifically, Theorem 3.1 from \cite{daskalakis2021near} (restated as Lemma~\ref{lemma:matrix_external_regret} in our Appendix~\ref{app:lemmas}) shows that with a learning rate $\eta_\tau = \Theta(\frac{1}{N\log^4 L_\tau})$, there exists a constant $C>1$ such that for any $(i, \tau, s,h)\in\mcn\times [\bar{\tau}]\times \mcs\times[H]$,
	\begin{equation}\label{eqn:c8}
		\operatorname{Reg}_{h, i}^\tau(s) \leq \frac{C N H \log A_i \cdot \log^4 L_\tau}{L_\tau}.
	\end{equation}
	Notice that we multiplied the regret bound by $H$ because  \cite{daskalakis2021near} assumes the rewards to be from $[0,1]$ but our rewards lie in $[0,H]$. 
	According to the definition in Algorithm~\ref{alg:certify_cce}, the behavior of the policy $\bar{\pi}_h^t$ is unchanged for all $t$ within the same stage $\tau$ as it always uniformly samples a time index from the previous stage and plays the corresponding history policy. 
	Consequently, the value estimation error $\zeta_h^t$ does not change within a stage $\tau(t)$; that is, $\zeta_h^t$ takes the same value for all $t\in[t_\tau^{\text{start}}, t_\tau^{\text{end}}]$. 
	We occasionally slightly abuse the notation and use $\zeta_h^\tau$ to denote the estimation error for a stage $\tau$.
	This immediately implies that $\frac{1}{L_{\tau-1}} \sum_{j=t_{\tau-1}^{\text{start}}}^{t_{\tau-1}^{\text{end}}} \zeta_{h+1}^j = \zeta_{h+1}^{\tau-1}.$
	Substituting \eqref{eqn:c8} and the above equation back to the recursion \eqref{eqn:c9}, we obtain that
	\begin{align}
		\zeta^t_h \leq &\zeta_{h+1}^{\tau-1} + \frac{C N H \log \amax \cdot \log^4 L_{\tau-1}}{L_{\tau-1}}\nonumber\\
		\leq& \sum_{h'=h}^H \frac{C N H \log \amax \cdot \log^4  T}{L_{\tau-h'+h-1}}\\
		\leq&  \frac{3C N H^2 \log \amax \cdot \log^4  T}{L_\tau},\label{eqn:c7}
	\end{align}
	where the second step is by applying the inequality recursively over $h$, and the last step holds because our choice of the stage lengths $L_{\tau+1}=\floor{(1+1/H)L_\tau}$ implies that
	\[
	\frac{1}{L_{\tau-h'+h-1}}\leq \frac{1}{L_\tau} \l 1+\frac{1}{H} \r^{h'-h+1} \leq \frac{1}{L_\tau} \l 1+\frac{1}{H} \r^H \leq \frac{3}{L_\tau}. 
	\]
	We then substitute \eqref{eqn:c7} back to \eqref{eqn:c6} and change the counting method to obtain
	\[
	\begin{aligned}
		\operatorname{CCE-gap}(\bar{\pi}) 
		\leq \frac{1}{T}\sum_{t=1}^T \zeta_1^t
		\leq &\frac{1}{T}\sum_{\tau=1}^{\bar{\tau}} \sum_{j=t_\tau^{\text{start}}}^{t_\tau^{\text{end}}}\frac{3C N H^2 \log \amax \cdot \log^4  T}{L_\tau} \\
		\leq& \frac{3C N \bar{\tau} H^2 \log \amax \cdot \log^4  T}{T}. 
	\end{aligned}
	\]
	It remains to bound the total number of stages $\bar{\tau}$. 
	Since the lengths of the stages increase exponentially as $L_{\tau+1}=\floor{(1+1/H)L_\tau}$ and the $\bar{\tau}$ stages sum up to $T$ iterations, by taking the sum of a geometric series, it suffices to find a value of $\bar{\tau}$ such that $(1+1/H)^{\bar{\tau}}\geq T/H$. 
	Using the Taylor series expansion, one can show that 
	$( 1 + \frac{1}{H})^H \geq e-\frac{e}{2H}$, and hence any $\bar{\tau}\geq \frac{H\log T}{\log (e/2)}$ satisfies the condition. 
	This completes the proof of the theorem.
\end{proof}

%% file: appendix_simulations.tex
\section{Simulation Details}\label{app:simulations}

\begin{table}[!t]
	\begin{minipage}{.5\linewidth}
		\centering
		\caption{Reward matrices for Player 1.}\label{tbl:player1}
		\begin{tabular}[b]{c|cc}
			$s_0$ & $b_0$ & $b_1$ \\
			\hline
			$a_0$ & 0.8 & 0.2 \\
			$a_1$ & 0.0 & 1.0
		\end{tabular}\qquad
		\begin{tabular}[b]{c|cc}
			$s_1$ & $b_0$ & $b_1$ \\
			\hline
			$a_0$ & 1.0 & 0.2 \\
			$a_1$ & 0.5 & 0.8
		\end{tabular}
	\end{minipage}%
	\begin{minipage}{.5\linewidth}
		\centering
		\caption{Reward matrices for Player 2.}\label{tbl:player2}
		\begin{tabular}[b]{c|cc}
			$s_0$ & $b_0$ & $b_1$ \\
			\hline
			$a_0$ & 0.2 & 1.0 \\
			$a_1$ & 0.5 & 0.0
		\end{tabular}\qquad
		\begin{tabular}[b]{c|cc}
			$s_1$ & $b_0$ & $b_1$ \\
			\hline
			$a_0$ & 0.5 & 1.0 \\
			$a_1$ & 1.0 & 0.2
		\end{tabular}
	\end{minipage} 
\end{table}
Our numerical studies are conducted on a simple general-sum Markov game with $2$ players, $2$ states $\mcs=\{s_0,s_1\}$ and $H=2$ steps per episode.
Each player has $2$ candidate actions $\mca=\{a_0,a_1\}$ and $\mcb=\{b_0,b_1\}$, respectively. 
The reward matrices for Player 1 and Player 2 at the two states are given in Tables~\ref{tbl:player1} and~\ref{tbl:player2}, respectively. 
The state transition function is defined as follows: In both states $s_0$ and $s_1$, if the two players take matching actions (namely $(a_0,b_0)$ or $(a_1,b_1)$), the system stays at the current state with probability 0.8, and transitions to the other state with probability 0.2. 
On the other hand, if the two players take opposite actions (namely $(a_0,b_1)$ or $(a_1,b_0)$), the environment will stay at the current state with probability 0.2, and will transition to the other state with probability 0.8. 
We choose a constant learning rate $\eta=0.2$ for all the three algorithms. 
We have also experimented with other choices of the transition and reward functions and have observed similar behavior as shown in Figures~\ref{fig:1} and~\ref{fig:2}.